\documentclass{article}
\usepackage[margin=1.5in]{geometry}
\usepackage{amsthm}
\usepackage{amsmath}
\usepackage{amssymb}
\usepackage{graphicx}
\usepackage{hyperref}
\usepackage{xcolor}
\usepackage{subcaption}
\usepackage{xspace}
\usepackage[title]{appendix}

\newtheorem{definition}{Definition}
\newtheorem{theorem}{Theorem}
\newtheorem{lemma}{Lemma}

\graphicspath{{./figure/}}
\newcommand{\etal}{\textit{et al}.\xspace}
\DeclareMathOperator{\polylog}{polylog}

\begin{document}

\title{Translation Invariant Fr\'echet Distance Queries}

\author{
 Joachim Gudmundsson, Andr{\'e} van Renssen, Zeinab Saeidi, Sampson Wong
}

\date{}

\maketitle

\begin{abstract}
The Fr\'echet distance is a popular similarity measure between curves. For some applications, it is desirable to match the curves under translation before computing the Fr\'echet distance between them. This variant is called the Translation Invariant Fr\'echet distance, and algorithms to compute it are well studied. The query version, finding an optimal placement in the plane for a query segment where the Fr\'echet distance becomes minimized, is much less well understood.

We study Translation Invariant Fr\'echet distance queries in a restricted setting of horizontal query segments. More specifically, we preprocess a trajectory in $\mathcal O(n^2 \log^2 n) $ time and $\mathcal O(n^{3/2})$ space, such that for any subtrajectory and any horizontal query segment we can compute their Translation Invariant Fr\'echet distance in $\mathcal O(\polylog n)$ time. We hope this will be a step towards answering Translation Invariant Fr\'echet queries between arbitrary trajectories.

\end{abstract}

\section{Introduction}
The Fr{\'e}chet distance is a popular measure of similarity between curves as it takes into account the location and ordering of the points along the curves, and it was introduced by Maurice Fr{\'e}chet in 1906~\cite{15}. Measuring the similarity between curves is an important problem in many areas of research, including computational geometry~\cite{1,b-wwdtt-14,9}, computational biology~\cite{JiangXZ08,WylieZ13}, data mining~\cite{12,13,wszsz-estsm-13}, image processing~\cite{2,14} and geographical information science~\mbox{\cite{l-cma-14,8,rt-hcmrp-14,td-tsm-15}}. 

The Fr{\'e}chet distance is most commonly described as the dog-leash distance. Let a trajectory be a polygonal curve in Euclidean space. Consider a man standing at the starting point of one trajectory and the dog at the starting point of another trajectory. A leash is required to connect the dog and its owner. Both the man and his dog are free to vary their speed, but they are not allowed to go backward along their trajectory. The cost of a walk is the maximum leash length required to connect the dog and its owner from the beginning to the end of their trajectories. The Fr{\'e}chet distance is the minimum length of the leash that is needed over all possible walks. More formally, for two curves $A$ and $B$ each having complexity $n$, the Fr{\'e}chet distance between $A$ and $B$ is defined as:
\begin{equation*}
\delta_{F} (A,B)=\inf_{\mu} \max_{a \in A} dist(a,\mu(a))
\end{equation*}
where $dist(a,b)$ denotes the Euclidean distance between point $a$ and $b$ and $\mu:A \rightarrow B$ is a continuous and non-decreasing function that maps every point in $a \in A$ to a point in $\mu(a) \in B$.

Since the early 90's the problem of computing the Fr{\'e}chet distance between two polygonal curves has received considerable attention. In 1992 Alt and Godau~\cite{1} were the first to consider the problem and gave an $\mathcal O(n^2 \log n)$ time algorithm for the problem. 
The only improvement since then is a randomized algorithm with running time $\mathcal O(n^2(\log\log n)^2)$ in the word RAM model by Buchin~\etal~\cite{bbmm-fswdi-17}. In 2014 Bringmann~\cite{b-wwdtt-14} showed that, conditional on the Strong Exponential Time Hypothesis (SETH),  there cannot exist an algorithm with running time $\mathcal O(n^{2-\varepsilon})$ for any $\varepsilon > 0$. 
Even for realistic models of input curves, such as $c$-packed curves~\cite{9}, where the total length of edges inside any ball is bounded by $c$ times the radius of the ball, exact distance computation requires $n^{2-o(1)}$ time under SETH~\cite{b-wwdtt-14}. Only by allowing a $(1 +\varepsilon)$-approximation can one obtain near-linear running times in $n$ and $c$ on $c$-packed curves~\cite{bk-iafdc-17,9}.

In some applications, it is desirable to match the two curves under translation before computing the Fr\'echet distance between them. For example, in sign language or in handwriting recognition, translating an entire movement pattern in space does not change the meaning of the pattern. Other applications where this is true include finding common movement patterns of athletes in sports, of animals in behavioural ecology, or to find similar proteins.

Formally, we match two polygonal curves $A$ and $B$ under the Fr\'echet distance by computing the translation $\tau$ so that the Fr\'echet distance between $A+\tau$ and $B$ is minimised. This variant is called the Translation Invariant Fr\'echet distance, and algorithms to compute it are well studied~\cite{AltKW01,bkn-fdutc-19,JiangXZ08,DBLP:phd/dnb/Wenk03}. Algorithms for the Translation Invariant Fr\'echet distance generally carry higher running times than for the standard Fr\'echet distance, moreover, these running times depend on the dimension of the input curves and on whether the discrete or continuous variant of the Fr\'echet distance is used.

For a discrete sequence of points in two dimensions, Bringmann~\etal~\cite{bkn-fdutc-19} recently provided an $\mathcal O(n^{14/3})$ time algorithm to compute the Translation Invariant Fr\'echet distance, and showed that the problem has a conditional lower bound of $\Omega(n^4)$ under SETH. For continuous polygonal curves in two dimensions, Alt~\etal~\cite{AltKW01} provided an $\mathcal O(n^8 \log n)$ time algorithm, and Wenk~\cite{DBLP:phd/dnb/Wenk03} extended this to an $\mathcal O(n^{11} \log n)$ time algorithm in three dimensions. If we allow for a $(1+\varepsilon)$-approximation then there is an $\mathcal O(n^2 / \varepsilon^2)$ time algorithm~\cite{AltKW01}, which matches conditional lower bound for approximating the standard Fr\'echet distance~\cite{b-wwdtt-14}. 

For both the standard Fr\'echet distance and the Translation Invariant Fr\'echet distance, subquadratic and subquartic time algorithms respectively are unlikely to exist under SETH~\cite{b-wwdtt-14,bkn-fdutc-19}. However, if at least one of the trajectories can be preprocessed, then the Fr\'echet distance can be computed much more efficiently. 

Querying the standard Fr\'echet distance between a given trajectory and a query trajectory has been studied~\cite{3,10,9,5,4}, but due to the difficult nature of the query problem, data structures only exist for answering a restricted class of queries. There are three results which are most relevant. The first is De~Berg~\etal's~\cite{10} data structure, which answers Fr\'echet distance queries between a horizontal query segment and a vertex-to-vertex subtrajectory of a preprocessed trajectory. Their data structure can be constructed in $\mathcal O(n^2 \log^2 n) $ time using $\mathcal O(n^2)$ space such that queries can be answered in $\mathcal O(\log^2 n)$ time. The second is a follow up paper by Buchin~\etal~\cite{eurocg2020improved}, which proves that the data structure of De~Berg~\etal's~\cite{10} requires only $\mathcal O(n^{3/2})$ space. The third is Driemel and Har-Peled's~\cite{9} data structure, which answers approximate Fr\'echet distance queries between a query trajectory of complexity $k$ and a vertex-to-vertex subtrajectory of a preprocessed trajectory. The data structure can be constructed in $ \mathcal{O}(n \log^3 n)$ using $\mathcal{O}(n \log n)$ space, and a constant factor approximation to the Fr\'echet distance can be answered in $\mathcal{O}(k^2 \log n \log (k \log n ))$ time. In the special case when $k=1$, the approximation ratio can be improved to $(1+\varepsilon)$ with no increase in preprocessing or query time in terms of $n$. New ideas are required for exact Fr\'echet distance queries on arbitrary query trajectories. Other query versions for the standard Fr\'echet distance have also been considered~\cite{3,5,4}.

Querying the Translation Invariant Fr\'echet distance is less well understood. This is not surprising given the complexity of computing the Translation Invariant Fr\'echet distance. Nevertheless, in our paper we are able to answer exact Translation Invariant Fr\'echet queries in a restricted setting of horizontal query segments. We hope this will be a step towards answering exact Translation Invariant Fr\'echet queries between arbitrary trajectories.

In this paper, we answer exact Translation Invariant Fr\'echet distance queries between a subtrajectory (not necessarily vertex-to-vertex) of a preprocessed trajectory and a horizontal query segment. The data structure can be constructed in $\mathcal O(n^2 \log^2 n) $ time using $\mathcal O(n^{3/2})$ space such that queries can be answered in $\mathcal O(\polylog n)$ time. We use Megiddo's parametric search technique~\cite{megiddo1981applying} to De~Berg~\etal's~\cite{10} data structure to optimise the Fr\'echet distance. We hope that as standard Fr\'echet distance queries become more well understood, similar optimisation methods could lead to improved data structures for the Translation Invariant Fr\'echet distance as well. 

\section{Preliminaries}
\label{sec: preliminaries}

Let 
$p_{1},\ldots ,p_{n}$ be a sequence of $n$ points in the plane. We denote 
$\pi=(p_{1},p_{2}\ldots,$ $p_{n})$
to be the polygonal curve defined by this sequence. Let $x_{0}\leq x_{1}$ and $y \in \mathbb{R}$, and define $p=\left(x_{0},y\right) $ and $q=\left( x_{1},y\right) $ so that $Q=pq$ is a horizontal segment in the plane. Let $u$ and $v$ be two points on the trajectory $\pi$, then from \cite{10}, the Fr{\'e}chet distance between $\pi[u,v]$ and $Q$ can be computed by using the formula:
\begin{equation*}
\delta_{F}(\pi[u,v],pq)=\max \lbrace \Vert up\Vert, \Vert vq\Vert, \delta_{\overrightarrow{h}} (\pi[u,v],pq), B(\pi[u,v],y)  \rbrace.
\label{1}
\end{equation*}

The first two terms are simply the distance between the starting points of the two trajectories, and the ending points of the two trajectories. The third term is the directed Hausdorff distance between $\pi[u,v]$ and $Q$ which can be computed from:
\begin{equation*}
\delta_{\overrightarrow{h}} (\pi[u,v], Q)=\max \lbrace \max \limits_{p_{i}.x \in (-\infty, x_{0}]} \Vert p-p_{i} \Vert, \max\limits_{p_{i}.x \in [x_{1},\infty)} \Vert q-p_{i} \Vert, \max_{i} \Vert y-p_{i}.y \Vert \rbrace,
\label{14}
\end{equation*}
where each $p_i$ in the formula above is a vertex of the subtrajectory $\pi[u,v]$, and $p_i.x$ denotes the $x$-coordinate of~$p_i$. The formula handles three cases for mapping every point of $\pi[u,v]$ to its closest point on $Q$. The first term describes mapping points of $\pi[u,v]$ to the left of $p$ to their closest point $p$. The second term describes mapping points of $\pi[u,v]$ to the right of $q$ analogously. The third term describes  mapping points of $\pi[u,v]$ that are in the vertical strip between $p$ and $q$ to their orthogonal projection onto $Q$. In later sections we refer to these three terms as $\delta_{\overrightarrow h}(L)$, $\delta_{\overrightarrow h}(R)$ and $\delta_{\overrightarrow h}(M)$ for the left, right, and middle terms of the Hausdorff distance respectively.

The fourth term in our formula for the Fr\'echet distance is the maximum backward pair distance over all backward pairs. A pair of vertices $(p_i, p_j)$ (with $j>i$) is a backward pair if $p_j$ lies to the left of $p_i$. The backward pair distance of $\pi[u,v]$ can be computed from:
\begin{equation*}
B(\pi[u,v],y)=\max_{\forall p_i,p_j \in \pi[u,v]: i \leq j, p_{i}.x \geq p_{j}.x} B_{\left( p_{i},p_{j}\right) }(y),
\end{equation*}
where $B_{(p_i,p_j)}(y)$ is the backward pair distance for a given backward pair $(p_i,p_j)$ and is defined as
\begin{equation*}
B_{\left( p_{i},p_{j}\right) }(y)=\min \limits_{x \in \mathbb{R}} \max \lbrace \Vert p_{i}-\left( x,y\right)  \Vert, \Vert p_{j}-\left( x,y\right)  \Vert \rbrace.
\end{equation*}

The distance terms in the braces compute the distance between a given point $(x,y)$ and the farthest of $p_i$ and $p_j$. Let us call this the backward pair distance of $(x,y)$. Then the function $B_{(p_i,p_j)}(y)$ denotes the minimum backward pair distance of a given backward pair $(p_i,p_j)$ over all points $(x,y)$ which have the same $y$-coordinate. Taking the maximum over all backward pairs gives us the backward pair distance for $\pi[u,v]$. Note that the backwards pair distance doesn't need a restriction to the $x$-coordinates of the horizontal segment and only depends on the $y$-coordinate of the horizontal segment.

In Figure~\ref{fig: backward_pair}, on the left, we show in red the point with minimum backward pair distance for $(p_i,p_j)$, for each $y$-coordinate. We show in red the associated distance for the minimum backwards pair distance on the right, where the distance is plotted along the $x$-axis. We see in the figure that the function $B_{(p_i,p_j)}(y)$ consists of two linear functions joined together in the middle with a hyperbolic function.

\begin{figure}[h]		
\centering
\includegraphics{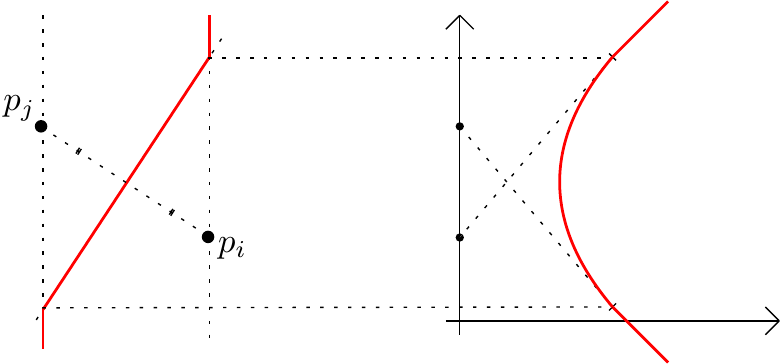}
\caption{For each $y$-coordinate, Left: the point with minimum backward pair distance, Right: the minimum backward pair distance.}
\label{fig: backward_pair}
\end{figure}

We extend the work of De Berg~\etal~\cite{10} in two ways. First, we provide a method for answering Fr\'echet distance queries between $\pi[u,v]$ and $Q$ when $u$ and $v$ are not necessarily vertices of $\pi$, and second, we optimise the placement of $Q$ to minimise its Fr\'echet distance to $\pi[u,v]$. We achieve both of these extensions by carefully applying Megiddo's parametric search technique~\cite{megiddo1981applying} to compute the optimal Fr\'echet distance. 

In order to apply parametric search, we are required to construct a set of critical values (which we will describe in detail at a later stage) so that an optimal solution is guaranteed to be contained within this set. Since this set of critical values is often large, we need to avoid computing the set explicitly, but instead design a decision algorithm that efficiently searches the set implicitly. Megiddo's parametric search~\cite{megiddo1981applying} states that if:
\begin{itemize}
    \item the set of critical values has polynomial size, and
    \item the Fr\'echet distance is convex with respect to the set of critical values, and
    \item a comparison-based decision algorithm decides if a given critical value is equal to, to the left of, or to the right of the optimum,
\end{itemize}
then there is an efficient algorithm to compute the optimal Fr\'echet distance in $\mathcal{O} (PT_{p}+T_{p}T_{s}\log P)$ time, where $P$ is the number of processors of the (parallel) algorithm, $T_p$ is the parallel running time and $T_s$ is the serial running time of the decision algorithm. For our purposes, $P=1$ since we run our queries serially, and $T_p = T_s$ = $\mathcal O(\polylog n)$ for the decision versions of our query algorithms.

\section{Computing the Fr\'echet Distance} \label{sec:vers 1}
The first problem we apply parametric search to is the following. Given any horizontal query segment $Q$ in the plane and any two points $u,v$ on $\pi$ (not necessarily vertices of $\pi$), determine the Fr{\'e}chet distance between $Q$ and the subtrajectory $\pi[u,v]$. 

\begin{figure}[bth]		
\centering
\includegraphics[width=0.8\textwidth]{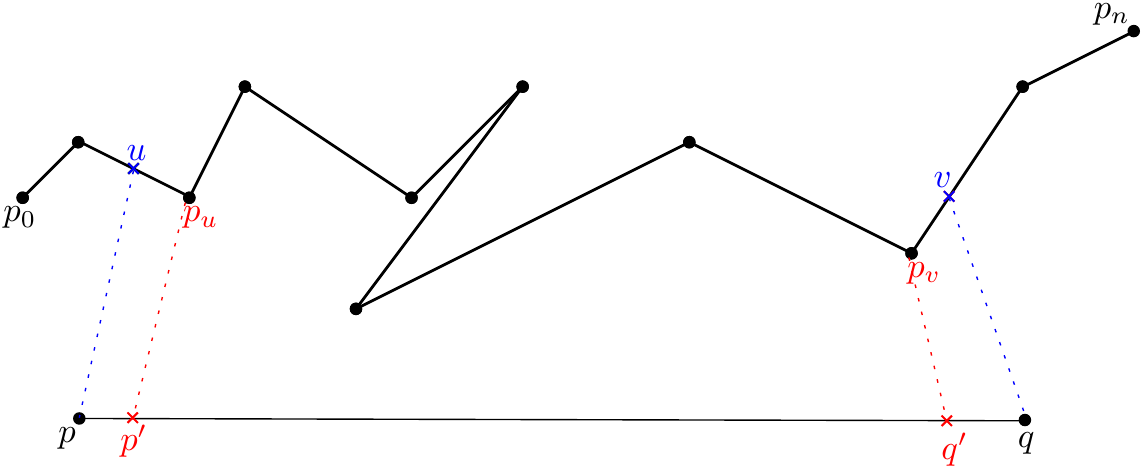}
\caption{The points $p'$ and $q'$ mapped to the vertices $p_u$ and $p_v$ of the trajectory.}
\label{2}
\end{figure}

Let $p_u$ be the first vertex of $\pi$ along $\pi[u,v]$ and let $p_v$ be the last vertex of $\pi$ along $\pi[u,v]$, as illustrated in Figure~\ref{2}. If $p_u$ and $p_v$ do not exist then $\pi[u,v]$ is a single segment so the Fr\'echet distance between $\pi[u,v]$ and $Q$ can be computed in constant time. Otherwise, our goal is to build a Fr\'echet mapping $\mu: \pi[u,v] \rightarrow Q$ which attains the optimal Fr\'echet distance. We build this mapping $\mu$ in several steps. Our first step is to compute points $p'$ and $q'$ on the horizontal segment $pq$ so that $p' = \mu(p_u)$ and $q' = \mu(p_v)$.

If the point $p'$ is computed correctly, then the mapping $p' \to p_u$ allows us to subdivide the Fr{\'e}chet computation into two parts without affecting the overall value of the Fr{\'e}chet distance. In other words, we obtain the following formula:
\begin{equation}
\delta_{F}(\pi[u,v],pq)=\max \lbrace \delta_{F}(up_{u},pp'),~\delta_{F}(\pi[p_{u},v],p'q) \rbrace
\label{3}
\end{equation}

We now apply the same argument to $p_v$. We compute $q'$ optimally on the horizontal segment $p'q$ optimally so that mapping $p_v \to q'$ does not increase the Fr{\'e}chet distance between the subtrajectory $\pi[p_u, v]$ and the truncated segment $p'q$. In other words, we have:
\begin{equation}
\delta_{F}(\pi[u,v],pq)=\max \lbrace \delta_{F}(up_{u},pp'),~\delta_{F}(\pi[p_{u},p_{v}],p'q'),~\delta_{F}(p_{v}v,q'q) \rbrace
\label{4}
\end{equation}

Now that $p_u$ and $p_v$ are vertices of $\pi$, \cite{10} provides an efficient data structure for computing the middle term $\delta_{F}(\pi[p_{u},p_{v}],p'q')$. The first and last terms have constant complexity and can be handled in constant time. All that remains is to compute the points $p'$ and $q'$ efficiently. 

\begin{theorem}
Given a trajectory $\pi$ with $n$ vertices in the plane. There is a data structure that uses $\mathcal{O}(n^{2}\log^{2} n)$ preprocessing time and $\mathcal O(n^{3/2})$ space, such that for any two points $u$ and $v$ on $\pi$ (not necessarily vertices of $\pi$) and any horizontal query segment $Q$ in the plane, one can determine the exact Fr{\'e}chet distance between $Q$ and the subtrajectory from $u$ to $v$ in $\mathcal{O}(\log^{8} n)$ time.
\label{26}
\end{theorem}
\begin{proof}
\textbf{Decision Algorithm.}
Let $S$ be the set of critical values (defined later in this proof), let $s$ be the current candidate for the point $p'$, and let $F(s) = \max \{\delta_{F}(ps, up_u), \delta_{F}(s q, \pi[p_u,v])\}$ be the minimum Fr\'echet distance between $pq$ and $\pi[u,v]$ subject to  $p_u$ being mapped to $s$. Our aim is to design a decision algorithm that runs in $\mathcal O(\log^4 n)$ time that decides whether the optimal $p'$ is equal to $s$, to the left of $s$ or to the right of $s$. This is equivalent to proving that all points to one side of $s$ cannot be the optimal $p'$ and may be discarded.

We use the Fr\'echet distance formula from Section~\ref{sec: preliminaries} to rewrite $F(s)=\max(\Vert up \Vert, \Vert vq \Vert,$ $\Vert p_{u} s \Vert,$ $\delta_{\overrightarrow{h}} (\pi[p_u,v],sq), B(\pi[p_u,v],y))$. Then we take several cases for which of these five terms attains the maximum value $F(s)$, and in each case we either deduce that $p' = s$ or all critical values to one side of $s$ may be discarded.

\begin{itemize}
    \item If $F(s) = \max(\Vert up \Vert, \Vert vq \Vert, B(\pi[p_u,v],y))$, then $p' = s$. We observe that none of the three terms on the right hand side of the equation depend on the position of $s$. Hence, $F(s) = \max(\Vert up \Vert, \Vert vq \Vert, B(\pi[p_u,v],y)) \leq F(p')$, and since $F(p')$ is the minimum possible value, $F(s) = F(p')$. We have found a valid candidate for $p'$ and can discard all other candidates in the set $S$.
    \item If $F(s) = \Vert p_u s \Vert$ and $p_u$ is to the right (left) of $s$ (see Figure~\ref{fig:theo1-1}), then $p'$ is to the right (left) of $s$. We will argue this for when $p_u$ is to the right of $s$, but an analogous argument can be used when $p_u$ is to the left. We observe that all points $t$ to the left of $s$ will now have $\Vert p_u t \Vert > \Vert p_u s \Vert$. Hence, $F(s) = \Vert p_u s \Vert < \Vert p_u t \Vert \leq F(t)$ for all points $t$ to the left of $s$, therefore all points to the left of $s$ may be discarded. 
    \begin{figure}[ht]
        \centering
        \includegraphics{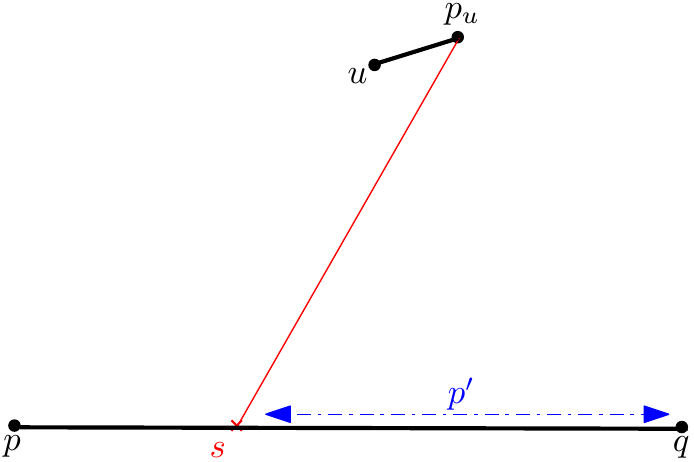}
        \caption{The decision algorithm moving to the right of the current candidate}
        \label{fig:theo1-1}
    \end{figure}
    
    \item If $F(s) = \delta_{\overrightarrow{h}} (\pi[p_u,v],sq)$, then $p'$ is to the left of $s$ (see Figure~\ref{fig:theo1-2}). The directed Hausdorff distance maps every point in $\pi[p_u,v]$ to their closest point on $sq$, so by shortening $sq$ to $tq$ for some point $t$ on $sq$ to the right of $s$, the directed Hausdorff distance cannot decrease. Hence, $F(s) \leq F(t)$ for all $t$ to the right of $s$, so all points to the right of $s$ may be discarded.
    \begin{figure}[ht]
        \centering
        \includegraphics{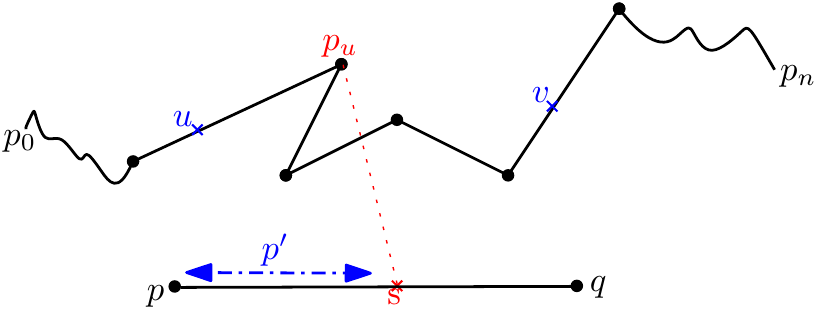}
        \caption{The decision algorithm moving to the left of the current candidate}
        \label{fig:theo1-2}
    \end{figure}
    
\end{itemize}

To determine $q'$ for a fixed candidate $s$ for $p'$, we treat the problem in a similar way. We consider the subtrajectory $\pi[p_u, v]$ and the horizontal line segment $sq$. Defining a function $G(t)$ representing the Fr\'echet distance when $p_v$ is mapped to $t$, we obtain a similar decision algorithm. The most notable difference is that since we now consider the end of the subtrajectory, the decisions for moving $t$ left and right are reversed. 

\textbf{Convexity.} We will prove that $F(s)$ is convex, and it will follow similarly that $G(t)$ is convex. It suffices to show that $F(s)$ is the maximum of convex functions, since the maximum of convex functions is itself convex. The three terms $\Vert up \Vert,$ $\Vert vq \Vert,$ $B(\pi[p_u,v],y)$ are constant. The term $\Vert p_u s \Vert$ is an upward hyperbola and is convex. If suffices to show that $\delta_{\overrightarrow{h}} (\pi[p_u,v],sq)$ is convex. 

We observe that the Hausdorff distance $\delta_{\overrightarrow{h}} (\pi[p_u,v],sq)$ must be attained at a vertex $p_i$ of $\pi[p_u,v]$, and that each of $\delta_{\overrightarrow{h}} (p_i,sq)$ as a function of $s$ is a constant function between $p$ and $p_i^*$, and a hyperbolic function between $p_i^*$ and $q$. Thus, the function for each $p_i$ is convex, so the overall Hausdorff distance function is also convex.

\textbf{Critical Values.} A critical value is a value $c$ which could feasibly attain the minimum value $F(c) = F(p')$. We represent $F(s)$ as the minimum of $n$ simple functions and then argue that the minimum of $F$ can only occur at the minimum of one of these functions, or at the intersection of a pair of these functions. 

First, $\Vert up \Vert, \Vert vq \Vert,  B(\pi[p_u,v],y)$ are constant functions in terms of $s$. Next, $\Vert p_{u}s \Vert$ is a hyperbolic function. Finally, $\delta_{\overrightarrow{h}} (\pi[p_u,v],sq)$ is not itself simple, but it can be rewritten as the combination of $n$ simple functions as described in the above section.

Hence, $F(s)$ is the combination (maximum) of $n$ simple functions, and these functions are simple in that they are piecewise constant or hyperbolic. Hence $F(s)$ attains its minimum either at the minimum of one of these $n$ functions, or at a point where two of these functions intersect. Therefore, there are at most $\mathcal O(n^2)$ critical values for $F(s)$.

\textbf{Query Complexity.} 
Computing $q'$ for a given candidate $s$ for $p'$ takes $\mathcal{O}(\log^{4} n)$ time: We can compute the terms $\Vert up \Vert$, $\Vert p_{u}s \Vert$, $\Vert vq \Vert$, and $\Vert p_{v} q' \Vert$ in constant time. The terms $B(\pi[p_u,p_v],y)$ and $\delta_{\overrightarrow{h}} (\pi[p_u,p_v],sq')$ can be computed in $\mathcal{O}(\log^{2} n)$ time using the existing data structure by De Berg~\etal~\cite{10}. 
We need to determine the time complexity of the sequential algorithm $T_{s}$, parallel algorithm $T_{p}$, and the number of the processor $P$. To find $q'$, the decision algorithm takes $T_{s}=\mathcal{O}(\log^{2} n)$. The parallel form runs on one processor in $T_{p}=\mathcal{O}(\log^{2} n)$. Substituting these values in the running time of the parametric search of $\mathcal{O} (PT_{p}+T_{p}T_{s}\log P)$ leads to $\mathcal{O}(\log^{4} n)$ time.

The above analysis implies that $p'$ itself can be computed in $\mathcal{O}(\log^{8} n)$ time: For a given $s$, the decision algorithm runs in $T_{s}=\mathcal{O}(\log^{4} n)$ as mentioned above. The parallel form of the decision algorithm runs on one processors in $T_{p}=\mathcal{O}(\log^{4} n)$. Substituting these values in the running time of the parametric search of $\mathcal{O} (PT_{p}+T_{p}T_{s}\log P)$ leads to $\mathcal{O}(\log^{8} n)$ time.

\textbf{Preprocessing and Space.}
To compute the second term of Formula~\ref{4}, we use the data structure by De Berg~\etal~\cite{10}. This data structure uses $\mathcal{O}(n^{2}\log^{2} n)$ preprocessing time and supports $\mathcal{O}(\log^{2} n)$ query time. Recently, the data structure was shown to require only $\mathcal O(n^{3/2})$ space~\cite{eurocg2020improved}.
\end{proof}

We note that the set of critical values can be restricted significantly, while still being guaranteed to contain optimal elements to use as $p'$ and $q'$. Specifically, we can reduce the size of this set from $\mathcal O(n^2)$ to $\mathcal O(n)$. Since this does not improve the running time of the above algorithm, details on this improvement are deferred to Appendix~\ref{app:improvement}.

\section{Minimizing the Fr\'echet Distance Under Vertical Translation} 
\label{sec:vers2}

We move on to the problem of minimising Fr\'echet distance under translations. We first focus on a special case where the horizontal segment can only be translated vertically. In Section~\ref{sec:vers3} we consider arbitrary translations of the horizontal segment.

To this end, let us consider the following problem. Let $\pi$ be a trajectory in the plane with $n$ vertices. We preprocess $\pi$ into a data structure such that for a query specified by 
\begin{enumerate}
\item two points $u$ and $v$ on the trajectory $\pi$,
\item two vertical lines $x_{1}$ and $x_{2}$ such that $\Vert x_{2}-x_{1} \Vert=L$ 
\end{enumerate}
one can quickly find a horizontal segment $l_{y}$ that spans the vertical strip between $x_{1} $ and $x_{2}$ such that the Fr{\'e}chet distance between $l_{y}$ and the subtrajectory $\pi[u,v]$ is minimised; see Figure~\ref{13}. 

\begin{figure}[h]		
\centering
\includegraphics{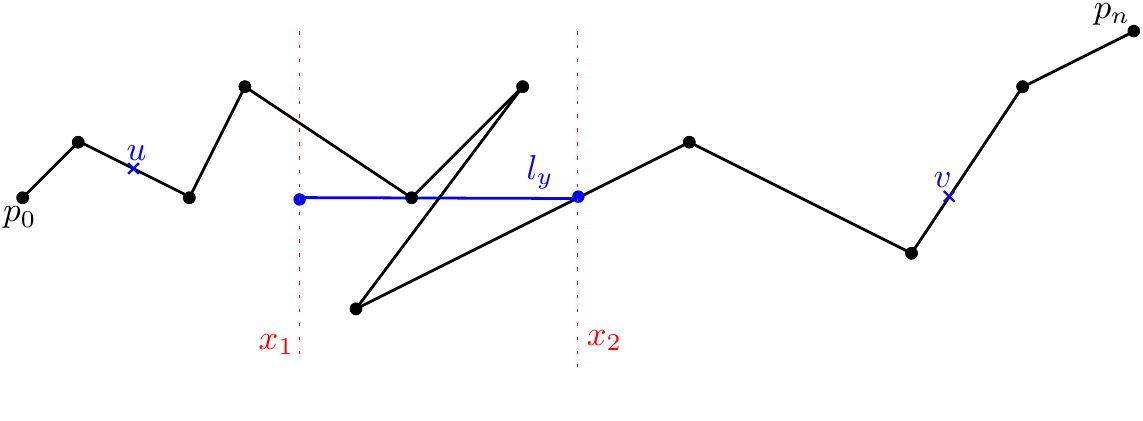}
\caption{Finding a horizontal segment $l_y$ in the vertical strip between $x_1$ and $x_2$ that minimises the Fr\'echet distance between $l_y$ and $\pi[u,v]$.}
\label{13}
\end{figure}

In the next theorem, we present a decision problem $D_{\pi[u,v]}(x_{1},x_{2},l_{y}^c)$ that, for a given trajectory $\pi$ with two points $u$ and $v$ on $\pi$ and two vertical lines $x = x_1$ and $x=x_2$, returns whether the line $l_y$ is above, below, or equal to the current candidate line $l_y^c$. We then use parametric search to find $l_y$ that minimises the Fr\'echet distance.

\begin{theorem} \label{thm:vers2}
Given a trajectory $\pi$ with $n$ vertices in the plane. There is a data structure that uses $\mathcal{O}(n^{2}\log^{2} n)$ preprocessing time and $\mathcal O(n^{3/2})$ space, such that for any two points $u$ and $v$ on $\pi$ (not necessarily vertices of $\pi$) and two vertical lines $x=x_1$ and $x=x_2$, one can determine the horizontal segment $l_y$ with left endpoint on $x=x_1$ and right endpoint on $x=x_2$ that minimises its Fr\'echet distance to the subtrajectory $\pi[u,v]$ in $\mathcal{O}(\log^{16} n)$ time. 
\end{theorem}
\begin{proof}
\textbf{Decision Algorithm.} Let $l_{y}^c$ be the current horizontal segment. To decide whether the line segment that minimises the Fr\'echet distance lies above or below $l_y^c$, we must compute the maximum of the terms that determine the Fr{\'e}chet distance: $\Vert up \Vert$, $\Vert vq \Vert$, $\delta_{\overrightarrow{h}} (\pi[u,v],pq)$, and $B(\pi[u,v],l_{y}^c)$. As mentioned in Section~\ref{sec: preliminaries}, we divide the directed Hausdorff distance into three different terms: $\delta_{\overrightarrow{h}}(L)$, $\delta_{\overrightarrow{h}}(R)$, and $\delta_{\overrightarrow{h}}(M)$. We first consider when one term determines the Fr\'echet distance, in which we have the following cases:

\begin{itemize}
\item $\Vert up \Vert$, $\Vert vq \Vert$, $\delta_{\overrightarrow{h}}(L)$, and $\delta_{\overrightarrow{h}}(R)$: Since the argument for these terms is analogous, we focus on $\Vert up \Vert$. If $u$ is located above $l_{y}^c$, the next candidate lies above $l_{y}^c$ (search continues above $l_{y}^c$). If $u$ lies below $l_{y}^c$, the next candidate lies below $l_{y}^c$ (search continues below $l_{y}^c$). If $u$ and $p$ have the same $y$-coordinate, we can stop, since moving $l_{y}^c$ either up or down increases the Fr{\'e}chet distance.
\item $B(\pi[u,v],l_{y}^c)$: If the midpoint of the segment between the backward pair determining the current Fr\'echet distance is located above $l_{y}^c$, the next candidate lies above $l_{y}^c$, since this is the only way to decrease the distance to the further of the two points of the backward pair. If this midpoint lies below $l_{y}^c$, the next candidate lies below $l_{y}^c$. If the midpoint is located on $l_{y}^c$, we can stop, because the term $B_{\left( p_{i},p_{j}\right) }(l_{y}^c)$ increases by either moving $l_{y}^c$ up or down.
\item $\delta_{\overrightarrow{h}}(M)$: If the point with maximum projected distance is located above $l_{y}^c$, the next candidate lies above~$l_{y}^c$. If the point is below $l_y^c$, the next candidate lies below~$l_{y}^c$. If the point is on $l_y^c$, then we stop, but unlike in the first case, this maximum term and the overall Fr\'echet distance must both be zero in this case.
\end{itemize}

If more than one term determine the current Fr\'echet distance, we must first determine the direction of the implied movement for each term. If this direction is the same, we move in that direction. If the directions are opposite, we can stop, because moving in either direction would increase the other maximum term resulting in a larger Fr\'echet distance.

\textbf{Convexity.} It suffices to show the Fr\'echet distance between $\pi[u,v]$ and $l_y^c$ as a function of $y$ is convex. We show that this function is the maximum of several convex functions, and therefore must be convex. The first two terms for computing the Fr\'echet distance are $\Vert up \Vert$ and $\Vert vq \Vert$, which are hyperbolic in terms of $y$. Similarly to the previous section, we handle each of the Hausdorff distances by splitting them up Hausdorff distances for each vertex $p_i$. The left and right Hausdorff distances $\delta_{\overrightarrow{h}}(L)$ and $\delta_{\overrightarrow{h}}(R)$ for a single vertex $p_i$ is a hyperbolic function. The middle Hausdorff distance $\delta_{\overrightarrow{h}}(M)$ for a single vertex $p_i$ is a shifted absolute value function. In all cases, Hausdorff distance for a single vertex is convex, so the overall Hausdorff distance is also convex. Finally, the backward pair distance $B(\pi[u,v],l_{y}^c)$ as a function of $y$ is shown by De Berg~\etal\cite{10} to be two rays joined together in the middle with a hyperbolic arc. It is easy to verify that this function is convex.

\textbf{Critical Values.} A horizontal segment $l_y^c$ is a critical value of a decision algorithm if the decision algorithm could feasibly return that $l_y^c = l_y$. These critical values are the $y$-coordinates of the intersection points of two hyperbolic functions for each combination of two terms of determining the Fr\'echet distance or the minimum point of the upper envelope of two such hyperbolic functions. Therefore, there are only a constant number of critical values for each two terms. Each term gives rise to  $\mathcal{O}(n^2)$ hyperbolic functions (specifically, $B(\pi[u,v],l_{y}^c)$ can be of size $\Theta(n^2)$ in the worst case). Thus, there are $\mathcal{O}(n^{4})$ critical values.

\textbf{Query Complexity.} The decision algorithm runs 
in $T_{s}=T_p=\mathcal{O}(\log^{8} n)$ time since we use Theorem~\ref{26} to compute the Fr{\'e}chet distance for a fixed $l_{y}^c$. Substituting this in the running time of the parametric search $\mathcal{O} (PT_{p}+T_{p}T_{s}\log P)$ leads to a query time of $\mathcal{O}(\log^{16} n)$.

\textbf{Preprocessing and Space.} 
Since we compute the Fr\'echet distance of the current candidate $l_y^c$ using Theorem~\ref{26}, we require $\mathcal{O}(n^{2}\log^{2} n)$ preprocessing time and $\mathcal O(n^{3/2})$ space.
\end{proof}


\section{Minimizing the Fr\'echet Distance for Arbitrary Placement} 
\label{sec:vers3}
Finally, we consider minimising the Fr\'echet distance of a horizontal segment under arbitrary placement. Let $\pi$ be a trajectory in the plane with $n$ vertices. We preprocess $\pi$ into a data structure such that for a query specified by two points $u$ and $v$ on $\pi$ and a positive real value $L$, one can quickly determine the horizontal segment $l$ of length $L$ such that the Fr{\'e}chet distance between $l$ and the subtrajectory $\pi[u,v]$ is minimised. 

In the following theorem, we present a decision problem $D_{\pi[u,v]}(L,x_1)$ that, for a given trajectory $\pi$ with two points $u$ and $v$ on $\pi$ and a length $L$ and an $x$-coordinate $x_1$, returns whether the line $l$ has its left endpoint to the left, on, or to the right of $x_1$. We then apply parametric search to this decision algorithm to find the horizontal segment $l$ of length $L$ with minimum Fr\'echet distance to $\pi[u,v]$.

\begin{theorem}
Given a trajectory $\pi$ with $n$ vertices in the plane. There is a data structure that uses $\mathcal{O}(n^{2}\log^{2} n)$ preprocessing time and $\mathcal O(n^{3/2})$ space, such that for any two points $u$ and $v$ on $\pi$ (not necessarily vertices of $\pi$) and a length $L$, one can determine the horizontal segment $l$ of length $L$ that minimises the Fr\'echet distance to $\pi[u,v]$ in $\mathcal{O}(\log^{32} n)$ time. 
\end{theorem}
\begin{proof}
\textbf{Decision Algorithm.} For the decision algorithm, we only need to decide whether $l^c$ should be moved to the left or right, with respect to its current position. We classify the terms that determine the Fr\'echet distance in two classes: 
\begin{itemize}
\item $C_{1}$: This class contains the terms whose value is determined by the distance from a point on $\pi(u,v)$ to $p$ or $q$. Hence, it consists of $\Vert up \Vert$, $\Vert vq \Vert$, $\delta_{\overrightarrow{h}}(R)$, and $\delta_{\overrightarrow{h}}(L)$. 
\item $C_{2}$: This class contains the terms whose value is determined by the distance from a point on $\pi(u,v)$ to the closest point on $pq$. Hence, it consists of $\delta_{\overrightarrow{h}}(M)$ and $B(\pi[u,v],l_{y})$. 
\end{itemize}
Next, we show how to decide whether the next candidate line segment lies to the left or right of $l^c$ (i.e., the $x$-coordinate of its left endpoint lies to the left or right of the left endpoint of $l^c$) for each case where $D_{\pi[u,v]}(x_{1},x_{2},l_{y})$ stops. 

We decide this by considering each $C_1$ and $C_2$ term and the restriction they place on the next candidate line segment $pq$. After we do this for each individual $C_1$ or $C_2$ term, we take the intersection of all these restrictions. If the intersection is empty, then our placement of $pq$ was optimal, and our decision algorithm stops. Otherwise we can either move $pq$ to the left or to the right to improve the Fr\'echet distance.

First, consider the $C_1$ terms. Let us assume for now that the $C_1$ term is the distance term $\Vert up \Vert$. Then in order to improve the Fr\'echet distance to $u$, we need to place the horizontal segment $pq$ in such a way that $p$ lies inside the \emph{open} disk centered at $u$ with radius equal to the current Fr\'echet distance $d$. A similar condition holds for the other $C_1$ terms: each defines a disk of radius $d$ and the point it maps to in the next candidate needs to lie inside this disk. 

Similarly, the $C_2$ terms define horizontal open half-planes. Consider the term $\delta_{\overrightarrow{h}}(M)$. This term is reduced when the vertical projection distance to the line segment is reduced. Hence, if the point defining this term lies above $l^c$, this term can be reduced by moving the line segment upward and thus the half-plane is the half-plane above $l^c$. An analogous statement holds if the point lies below $l^c$. For the term $B(\pi[u,v],l_{y})$, we need to consider the midpoint of the bisector, since the implied Fr\'echet distance is the distance from $l^c$ to the further of the two points defining the bisector. Thus, the half-plane that improves the Fr\'echet distance is the one that lies on the same side of $l^c$ as this midpoint. 

To combine all the terms we do the following: First, we take all disks induced by the $C_1$ terms whose distance is with respect to $q$ and translate them horizontally to the left by a distance of $L$. This ensures that the disks constructed with respect to $p$ can now be intersected with the disks constructed with respect to $q$. We take the intersection of all $C_1$ and $C_2$ terms that defined the stopping condition of the vertical optimisation step. If this intersection is empty, by construction there is no point where we can move $p$ to in order to reduce the Fr\'echet distance. If it is not empty, we will show that it lies entirely to the left or entirely to the right of $p$ and thus implies the direction in which the next candidate lies. 

Now that we have described our general approach, we show which cases can occur and show that for each of them we can determine in which direction to continue (if any). 

\begin{figure}[bth]
\centering
\begin{subfigure}[t]{0.4\textwidth}
\centering
\includegraphics{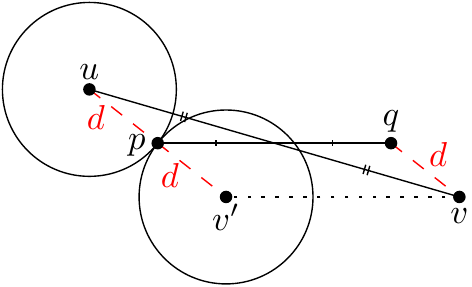}
\caption{The midpoint of $pq$ is the midpoint of $uv$.}
\end{subfigure}
\hspace{2em}
\begin{subfigure}[t]{0.5\textwidth}
\centering
\includegraphics{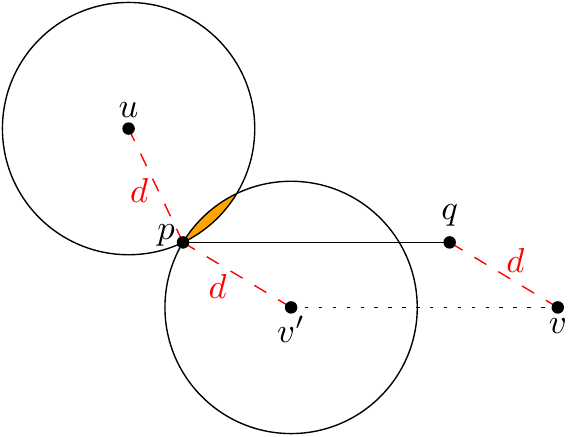}
\caption{Moving the midpoint of $pq$ towards the midpoint of $uv$.} 
\end{subfigure}
\caption{Determining where $l^c$ should be moved to reduce the Fr\'echet distance.}
\label{fig:moving_to_midpoint}
\end{figure}

\textit{Case 1.} $D_{\pi[u,v]}(x_{1},x_{2},l_{y})$ stops because of terms in $C_1$. If only a single term of $C_1$ is involved, say $\Vert up \Vert$, this implies that the $y$-coordinate of $u$ is the same as that of $l^c$ and thus its disk lies entirely to the left of $p$. Hence, we can reduce the Fr\'echet distance by moving $l^c$ horizontally towards $u$ and thus we pick our next candidate in that direction. The same argument follows analogously the $C_1$ term is $\Vert vq \Vert$, $\delta_{\overrightarrow{h}}(R)$, or $\delta_{\overrightarrow{h}}(L)$, the same argument follows analogously the distance is between a point on the trajectory 

If two terms of $C_1$ are involved, say $\Vert up \Vert$ and $\Vert vq \Vert$, their intersection can be empty (see Figure~\ref{fig:moving_to_midpoint}(a)) or non-empty (see Figure~\ref{fig:moving_to_midpoint}(b)). If it is empty, the midpoint of $pq$ is the same as the midpoint of $uv$, which implies that we cannot reduce the Fr\'echet distance. If the intersection is not empty, moving the endpoint of the line segment into this region potentially reduces the Fr\'echet distance. We note that since $\Vert up \Vert$ and $\Vert vq \Vert$ stopped the vertical optimisation, they lie on opposite sides of $l^c$. Hence, the intersection of their disks lies entirely to the left or entirely to the right of $p$ and thus determines in which direction the next candidate lies. 

If three terms in $C_1$ are involved, we again construct the intersection as described earlier. If this intersection is empty (see Figure~\ref{fig:C1C1C1Stopping}), we are again done. If it is not (see Figure~\ref{fig:C1C1C1}), it again determines the direction in which the our next candidate lies, as the intersection of three disks is a subset of the intersection of two disks.

\begin{figure}[h]
\centering
\begin{subfigure}[t]{0.5\textwidth}
\centering
\includegraphics{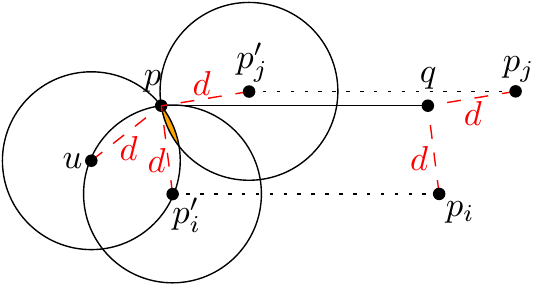}
\caption{The case where we can improve the Fr\'echet distance.}
\label{fig:C1C1C1}
\end{subfigure}
\hspace{1em}
\begin{subfigure}[t]{0.45\textwidth}
\centering
\includegraphics{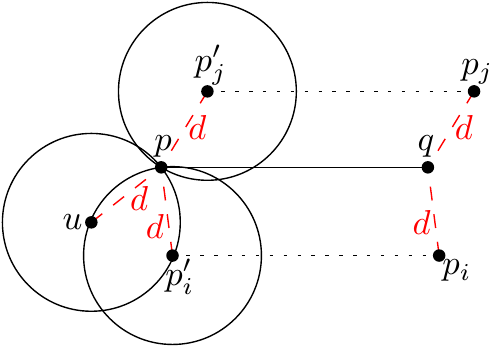}
\caption{The case where we cannot improve the Fr\'echet distance.}
\label{fig:C1C1C1Stopping}
\end{subfigure}
\caption{The case where we have three $C_1$ terms.}
\end{figure}

If there are more than three $C_1$ terms, we reduce this to the case of three $C_1$ terms. If the intersection of these disks is non-empty, then trivially the intersection of a subset of three of them is also non-empty. If the intersection is empty, we select a subset of three whose intersection is also empty. The three disks can be chosen as follows.  Insert the disks in some order and stop when the intersection first becomes empty. The set of three disks consists of the last inserted disk and the two extreme disks among the previously inserted disks. Since the boundary of all the disks must go through a single point and the disks have equal radius, these three disks will have an empty intersection. Hence, the case of more than three disks reduces to the case of three disks. 

\textit{Case 2.} $D_{\pi[u,v]}(x_{1},x_{2},l_{y})$ stops because of a term in $C_2$. Since the vertical optimisation stopped, we know that at least two $C_2$ terms are involved and there exists a pair that lies on opposite sides of $l^c$. These two terms define open half-planes whose intersection is empty, hence we cannot reduce the Fr\'echet distance further. 

\textit{Case 3.} $D_{\pi[u,v]}(x_{1},x_{2},l_{y})$ stops because of terms in $C_1$ and terms in $C_2$. If there are at least three terms in $C_1$, then we can ignore the $C_2$ terms and use the analysis provided in Case~1 on the three terms in $C_1$. If there are at least two terms in $C_2$, then we can ignore the $C_1$ terms and use the same analysis provided in Case~2 on the two terms in $C_2$. Therefore, without loss of generality, we can assume that there are at most two $C_1$ terms and at most one $C_2$ terms.

\begin{figure}[bth]
\centering
\begin{subfigure}[t]{0.45\textwidth}
\centering
\includegraphics{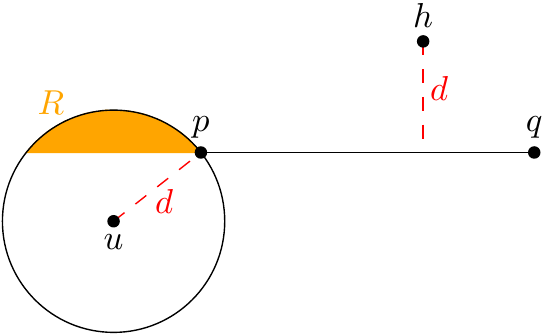}
\caption{The case where the $C_2$ term is $\delta_{\overrightarrow{h}}(M)$.}
\label{21}
\end{subfigure}
\hspace{2em}
\begin{subfigure}[t]{0.45\textwidth}
\centering
\includegraphics{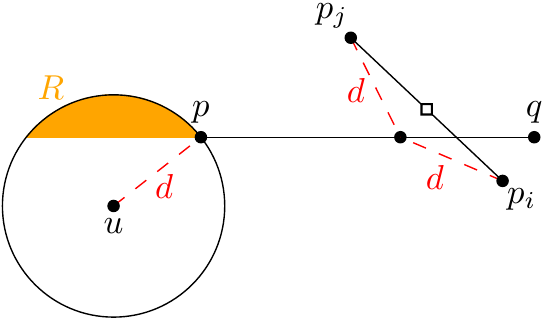}
\caption{The case where the $C_2$ term is $\max \limits_{{u}\leq i \leq j \leq v, p_{i}.x \geq p_{j}.x} B_{\left( p_{i},p_{j}\right) }(l^c)$.}
\label{22}
\end{subfigure}
\caption{Reduce the Fr\'echet distance when it is determined by a term of $C_1$ and a term of $C_2$.}
\end{figure}

The $C_2$ term can be either $\delta_{\overrightarrow{h}}(M)$ (see Figure~\ref{21}, where $h$ is the point at distance $d$) or $B(\pi[u,v],l_{y})$ (see Figure~\ref{22}, where $(p_{i},p_{j})$ is the backward pair with distance $d$). The region $R$ shows the intersection of the disk of a single $C_1$ term and the $C_2$ term. We note that since the point of the $C_1$ term and the point of the $C_2$ term lie on opposite sides of $l^c$, this intersection lies either entirely to the left or entirely to the right of $p$ or $q$, determining the direction in which our next candidate must lie.

The same procedure can be applied when there are two $C_1$ terms and using similar arguments, it can be shown that if the intersection is not empty, the direction to improve the Fr\'echet distance is uniquely determined. 

\textbf{Convexity.} Next, we show that $D_{\pi[u,v]}(L,x_1)$ is a convex function with respect to the parameter $x_1$. Let $l_{y}^c$ be the current horizontal segment and assume without loss of generality that the decision algorithm moves right to a new segment $l_{y'}$; see Figure~\ref{23}. Consider a linear interpolation from $l_{y}^c$ to $l_{y'}$. Let $l_{y''}$ be the segment at the midpoint of this linear interpolation. Since~$D_{\pi[u,v]}(L,x_1)$ is a continuous function, for continuous functions, convex is the same as midpoint convex, this implies that we only need to show that~$D_{\pi[u,v]}(L,x_1)$ is midpoint convex. 

Consider the two mappings that minimise the Fr{\'e}chet distance between $\pi\left[ u,v\right]$ and the horizontal segments $l_{y}^c$ and $l_{y'}$. Let $r$ be any point on~$\pi\left[ u,v\right]$ and let $a$ and $c$ be the points where $r$ is mapped to on $l_{y}$ and $l_{y'}$. 
Construct a point $b$ on $l_{y''}$ where $r$ will be mapped to by linearly interpolating $a$ and $c$. Performing this transformation for every point on $\pi[u,v]$, we obtain a valid mapping for $l_{y''}$, though not necessarily one of minimum Fr\'echet distance. 

We bound the distance between $r$ and $b$ in terms of $\Vert ra \Vert$ and $\Vert rc \Vert$. Consider the parallelogram consisting of $a$, $r$, $b$, and a point $r'$ that is distance $\Vert ra \Vert$ from $c$ and distance $\Vert rc \Vert$ from $a$; see Figure~\ref{24}. Since~$b$ is the midpoint of~$ac$, it is also the midpoint of $rr'$ in this parallelogram. We can conclude that $\Vert rb \Vert \leq (\Vert ra \Vert+\Vert rc \Vert)/2$ in this mapping. 

Since this property holds for any point $r$ on $\pi[u,v]$ and the Fr\'echet distance is the minimum over all possible mappings, the Fr\'echet distance of $l_{y''}$ is upper bounded by the average of the Fr\'echet distances of $l_{y}^c$ and $l_{y'}$. Therefore, the decision problem is convex. 

\begin{figure}[bth]
\centering
\begin{subfigure}[t]{0.55\textwidth}
\centering
\includegraphics{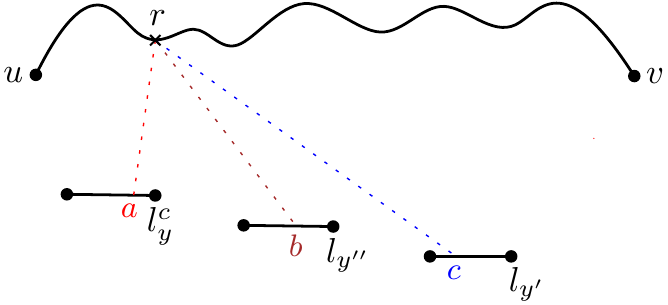}
\caption{Point $r$ is mapped to the three line segments.}
\label{23}
\end{subfigure}
\hspace{2em}
\begin{subfigure}[t]{0.3\textwidth}
\centering
\includegraphics{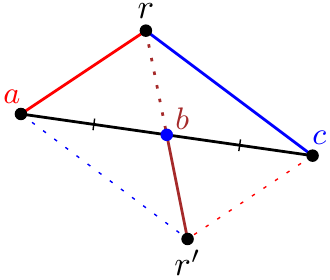}
\caption{Upper bounding $\Vert rb \Vert$.}
\label{24}
\end{subfigure}
\caption{The decision algorithm is a convex function with respect to the left endpoint of the line segment.}
\end{figure}

\textbf{Critical Values.} An $x$-coordinate $x_1$ is a critical value of a decision algorithm if the decision algorithm could feasibly return that the left endpoint of $l$ has $x$-coordinate $x_1$. 

For the $C_{1}$ class, these critical values are determined by up to three $C_1$ terms: the vertices themselves, the midpoint of any pair of vertices, and the center of the circle through the three (translated) points determining the Fr\'echet distance. Since each term in $C_1$ consists of at most $n$ points, there are $\mathcal{O}(n^{3})$ critical values in \textit{Case 1}. 

For the $C_{2}$ class, these critical values are the $x$-coordinates of the intersection points and minima of two hyperbolic functions, one for each element of each pair of two terms. 
Therefore, there are only a constant number of critical values for each two terms.  Each term gives rise to at most $\mathcal{O}(n^{2})$ hyperbolic functions (specifically, $B(\pi[u,v],l_{y})$ can be of size $\Theta(n^2)$ in the worst case). Thus, there are at most $\mathcal{O}(n^{4})$ critical values in \textit{Case 2}.

Using similar arguments, it can be shown that there are at most $\mathcal{O}(n^4)$ critical values in \textit{Case 3}, as they consist of at most two $C_1$ terms and at most one $C_2$ term. 

\textbf{Query Complexity.} The decision algorithm runs in $T_{s}=\mathcal{O}(\log^{16} n)$ time since we use Theorem~\ref{thm:vers2} to compute the optimal placement for a fixed left endpoint. The parallel form of the decision algorithm runs on one processor in $T_{p}=\mathcal{O}(\log^{16} n)$ time. Substituting these values in the running time of the parametric search of $\mathcal{O} (PT_{p}+T_{p}T_{s}\log P)$ leads to $\mathcal{O}(\log^{32} n)$ time.

\textbf{Preprocessing and Space.} Since we use the algorithm of Theorem~\ref{thm:vers2} to the optimal placement of $l^c$ for a given $x$-coordinate of its left endpoint, this requires $\mathcal{O}(n^{2}\log^{2} n)$ preprocessing time and $\mathcal O(n^{3/2})$ space. 
 \end{proof}

\section{Conclusion}

In this paper, we answered Translation Invariant Frechet distance queries between a horizontal query segment and a subtrajectory of a preprocessed trajectory. The most closely related result is that of De Berg~\etal~\cite{10}, which computes the normal Fr\'echet distance between a subtrajectory and a horizontal query segment. We extended this work in two ways. Firstly, we considered all subtrajectories, not just vertex-to-vertex subtrajectories. Secondly, we computed the optimal translation for minimising the Fr\'echet distance, thus our approach allowed us to compute both the normal Fr\'echet distance and the Translation Invariant Fr\'echet distance. All our queries can be answered in polylogarithmic time.

In terms of future work, one avenue would be to improve the query times. While our approach has polylogarithmic query time, the $\mathcal{O}(\log^{32} n)$ time needed for querying the optimal placement under translation is far from practical. Furthermore, reducing the preprocessing time or space of the data structure would this would make the approach more appealing. 

Other future work takes the form of generalising our queries further. In our most general form, we still work with a fixed length line segment with a fixed orientation. An interesting open problem is to see if we can also determine the optimal length of the line segment efficiently at query time. Allowing the line segment to have an arbitrary orientation seems a difficult problem to generalise our approach to, since the data structures we use assume that the line segment is horizontal. This can be extended to accommodate a constant number of orientations instead, but to extend this to truly arbitrary orientations, given at query time, will require significant modifications and novel ideas.


\bibliographystyle{plain}
\bibliography{main.bib} 

\newpage

\begin{appendices}
\section{Improving the Number of Critical Values}
\label{app:improvement}
In order to show that it suffices to use a set of critical values of size $\mathcal{O}(n)$ instead of $\mathcal{O}(n^2)$ to compute $p'$ and $q'$, we look more formally at what property a candidate needs to satisfy. 

\begin{definition}
\label{defn: represents_mu}
A point $s$ represents $p_u$ if and only if there exists a non-decreasing continuous mapping $\mu: \pi[u,v] \to pq$ such that $\mu$ achieves the Fr{\'e}chet distance and $\mu(p_u) = s$.
\end{definition}

Now we define a collection of points on $pq$ that could feasibly be representatives.

\begin{definition}
Given any vertex $p_i$ on the subtrajectory $\pi[u,v]$, let $p^*_i$ be the orthogonal projection of vertex $p_i$ onto the horizontal segment $pq$.
\end{definition}

\begin{definition}
Given any two vertices $p_i$ and $p_j$ on the subtrajectory $\pi[u,v]$, let $P_{ij}$ be the perpendicular bisector of $p_i$ and $p_j$. Let $P_{ij}^*$ be the intersection of the perpendicular bisector $P_{ij}$ with the horizontal segment $pq$.
\end{definition}

We now have all we need in place to define our set $S$ of candidates for $p'$ and $q'$. 

\begin{definition}
\label{defn: candidate_S}
Let $S$ be the set containing the following elements:
\begin{enumerate}
    \item the points $p$ and $q$, 
    \item all orthogonal projection points $p^*_i$, and
    \item all perpendicular bisector intersection points $P_{ij}^*$.
\end{enumerate}
\end{definition}

It now suffices to show that $S$ contains at least one representative for $p_u$. An analogous argument shows that $S$ contains a representative of $p_v$ as well. 

\begin{lemma}
There exists an element $s \in S$ on $pq$ that represents $p_u$.
\label{lem:element4p_u}
\end{lemma}
\begin{proof}
Assume for the sake of contradiction that there is no element $s \in S$ which represents $p_u$. Consider a mapping $\mu$ that achieves the Fr{\'e}chet distance and consider the point $\mu(p_u)$ on the horizontal segment $pq$. Since $\mu(p_u)$ represents $p_u$, $\mu(p_u)$ cannot be in $S$ and must lie strictly between two consecutive elements of $S$, say $s_L$ to its left and $s_R$ to its right (see Figure~\ref{fig:slsr}). Note that it may be the case that $s_L = p$ or $s_R = q$. Since $s_L$ and $s_R$ are elements of $S$, neither can represent $p_u$. Next, we reason about the implications of $s_L$ and $s_R$ not being able to represent $p_u$, before putting these together to obtain a contradiction.

\begin{figure}[h]
    \centering
    \includegraphics{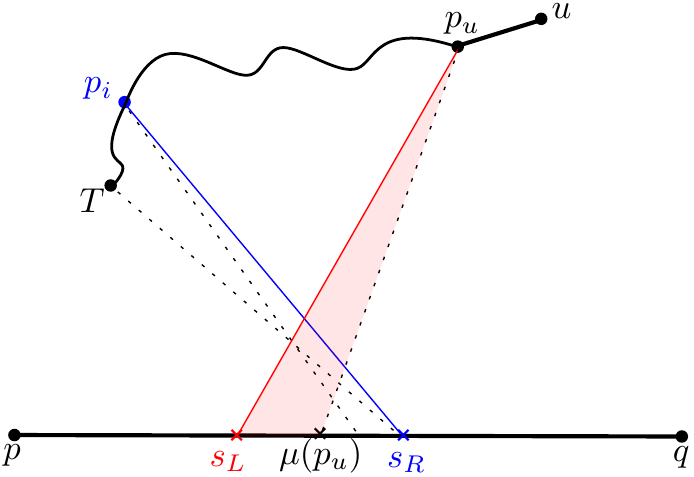}
    \caption{The point $\mu(p_u)$ lies between two consecutive elements $s_L$ and $s_R$. Distances that are greater than $d$ are thin solid and distances that are at most $d$ are dotted, where $d$ is the Fr{\'e}chet distance.}
    \label{fig:slsr}
\end{figure}

\textbf{$\boldsymbol{s_L}$ cannot represent $\boldsymbol{p_u}$.} This means that no mapping which sends $p_u \to s_L$ achieves the Fr{\'e}chet distance. Let us take the mapping $\mu$ and modify it into a new mapping $\mu_L$ in such a way that $\mu_L(p_u) = s_L$. We can do so by starting out parametrising $\mu_L$ with a constant speed mapping which sends $u \to p$ and $p_u \to s_L$. Next, we stay fixed at $p_u$ along the subtrajectory and move along the horizontal segment from $s_L$ to $\mu(p_u)$. The red shaded region in Figure~\ref{fig:slsr} describes this portion of the remapping. Now that $\mu_L(p_u) = \mu(p_u)$, we can use the original mapping for the rest. 

Since our new mapping $\mu_L$ maps $p_u$ to an element of $S$ that cannot represent it, we know that our modification must increase the Fr{\'e}chet distance. The only place where the Fr{\'e}chet distance could have increased is at the line segments where the mapping was changed and here $\mu_L(p_u) = s_L$ maximises the Fr\'echet distance. Hence, we have $ \Vert p_u s_L \Vert  > d$, where $d$ is the Fr{\'echet} distance, as shown in Figure~\ref{fig:slsr}. But $ \Vert p_u \mu(p_u) \Vert  \leq d$, so we can deduce that $p_u$ is closer to $\mu(p_u)$ than $s_L$. Therefore, $p_u$ is to the right of $s_L$. Finally, if $s_L$ and $s_R$ were on opposite sides of $p_u^*$, then $s_L$ and $s_R$ would not be consecutive, therefore $p_u$ must be on the same side of $s_L$ and $s_R$. Therefore, $p_u$ is to the right of the entire segment $s_L s_R$.

\textbf{$\boldsymbol{s_R}$ cannot represent $\boldsymbol{p_u}$.} Again, no mapping which sends $p_u \to s_R$ achieves the Fr{\'e}chet distance, so we use the same approach and modify $\mu$ into a new mapping mapping $\mu_R$ in such a way that $\mu_R(p_u) = s_R$. To this end, we keep the mapping $\mu_R$ the same as $\mu$ until it reaches $p_u$, and then while staying at $p_u$, we fastforward the movement from $\mu(p_u)$ along the horizontal segment so that $\mu_R(p_u) = s_R$. Next, we stay at $s_R$ and fastforward the movement along the subtrajectory, until we reach the first point $T$ on the subtrajectory such that $\mu(T) = s_R$ in the original mapping. From point $T$ onwards we can use the original mapping $\mu$. 

Since our new mapping $\mu_R$ maps $p_u$ to an element of $S$ that does not represent it, we cannot have achieved the Fr{\'e}chet distance. The first change we applied was staying at $p_u$ and fastforwarding the movement from $\mu(p_u)$ to $s_R$. However, since we know from above that $p_u$ is to the right of the entire segment $s_L s_R$, this fastforwarding moves closer to $p_u$, so this part cannot increase the Fr{\'e}chet distance. The second change we applied, staying at $s_R$ and fastforwarding the movement from $p_u$ to $T$, must therefore be the change that increases the Fr{\'e}chet distance. Thus, there must be a point on the subtrajectory $\pi[p_u,T]$ which has distance greater than $d$, the Fr{\'e}chet distance, to the point $s_R$. Since the distance to a point $s_R$ is maximal at vertices of $\pi[p_u,T]$, we can assume without loss of generality that $ \Vert p_i s_R \Vert  > d$ for some vertex $p_i$. Consider $\mu(p_i)$ in the original mapping. Since $p_i$ is on the subtrajectory $\pi[p_u,T]$, $\mu(p_i)$ must be between $\mu(p_u)$ and $\mu(T) = s_R$. This mapping of $p_i$ to $\mu(p_i)$ is shown as a black dotted line in Figure~\ref{fig:slsr}. Using a similar logic as before, $ \Vert p_i \mu(p_i) \Vert  \leq d$ and $ \Vert p_i s_R \Vert  > d$, so $p_i$ must lie to left of $s_R$. And since $s_L$ and $s_R$ are consecutive elements of $S$, we deduce that $p_i$ is to the left of the entire segment $s_L s_R$.

\textbf{Putting these together.} We now have the full diagram as shown in Figure~\ref{fig:slsr}. The vertex $p_u$ is to the right of both $s_L$ and $s_R$ and the vertex $p_i$ is to the left of both $s_L$ and $s_R$. We also have inferred that $ \Vert p_u s_L \Vert  > d$ and $ \Vert p_i s_R \Vert  > d$. Moreover, since $ \Vert p_u \mu(p_u) \Vert  \leq d$ and $ \Vert p_i \mu(p_i) \Vert  \leq d$, we also have that $ \Vert p_u s_R \Vert  \leq d$ and $ \Vert p_i s_L \Vert  \leq d$, since this just moves these endpoints closer to $p_u$ and $p_i$ respectively. 

Finally, we will show that $P_{ui}^*$ lies between $s_L$ and $s_R$, reaching the intended contradiction. We do so by considering the function $f(x) =  \Vert x p_u \Vert  -  \Vert x p_i \Vert $ for all points $x$ between $s_L$ and $s_R$. From our length conditions, we have that $f(s_L) > 0$, $f(s_R) < 0$. Furthermore, since $f(x)$ is a continuous function, by the intermediate value theorem, there is a point $x$ strictly between $s_L$ and $s_R$ such that $f(x) = 0$. Since $f(x) = 0$, the point $x$ is equidistant from $p_u$ and $p_i$ so therefore lies on both $P_{ui}$ and the horizontal segment $pq$. Therefore $x = P_{ui}^*$ and is an element of $S$ between two consecutive elements $s_L$ and $s_R$, giving us a contradiction.
 \end{proof}

Note that in the above proof, we require only $P_{ui}^*$ to be in the candidate set when we are computing $p'$, and also only when $(p_u, p_i)$ is a backward pair. This means that for computing $p'$ and $q'$ respectively, we only require the bisector intersections $P_{ui}^*$ and $P_{jv}^*$ to be in $S$, hence reducing the size of $S$ from $\mathcal O(n^2)$ to $\mathcal O(n)$. 
\end{appendices}

\end{document}